\documentclass[11pt]{article}
\usepackage{currfile}
\usepackage{color}
\usepackage{mathrsfs}
\usepackage{graphicx}
\usepackage{caption}
\usepackage{subfigure}
\usepackage{amsmath,amssymb}
\usepackage[citecolor=blue]{hyperref}

\newtheorem{theorem}{Theorem}[section]

\newtheorem{conjecture}[theorem]{Conjecture}
\newtheorem{problem}[theorem]{Problem}

\newtheorem{proposition}[theorem]{Proposition}
\newtheorem{claim}{Claim}

\newenvironment{proof}{\noindent {\bf Proof.}}{\rule{3mm}{3mm}\par\medskip}

\newcommand{\ENDproof}{~\rule{3mm}{3mm}\medskip\par}

\title{Extremal  numbers of disjoint triangles in $r$-partite graphs}

\author{\small   Junxue Zhang\footnote{Corresponding author} \\ {\small  Center for Combinatorics and LPMC, Nankai University, Tianjin 300071, China}\\
{\small Email: jxuezhang@163.com}}

\date{}
\begin{document}\maketitle
\begin{abstract}
For two graphs $G$ and $F$, the extremal number of $F$ in $G$, denoted by  {ex}$(G,F)$,  is the maximum number of edges in a spanning subgraph of $G$ not  containing $F$ as a subgraph. 
  Determining {ex}$(K_n,F)$ for a given graph $F$
   is a classical extremal problem  in graph theory.   
 In 1962, 	Erd\H{o}s  determined {ex}$(K_n,kK_3)$, which generalized Mantel's Theorem. 	  
On the other hand, in 1974,  {Bollob\'{a}s}, Erd\H{o}s, and Straus 
	  determined {ex}$(K_{n_1,n_2,\dots,n_r},K_t)$, which  extended Tur\'{a}n's Theorem to complete multipartite graphs. 
{ In this paper,}
 we determine {ex}$(K_{n_1,n_2,\dots,n_r},kK_3)$ for {$r\ge 4$ and  $10k-4\le n_1+4k\le n_2\le n_3\le \cdots \le n_r$.  }


	\end{abstract}

\section{Introduction}
All graphs  considered in this paper are finite and simple. 
 We follow \cite{WD2001} for undefined notation and terminology.  Let $G=G(V(G),E(G))$ be a graph, where $V(G)$ is the vertex set  and $E(G)$ is the edge set.  Denote $e(G)=|E(G)|$.
 For any subset $S\subseteq V(G)$, we use   $G[S]$ to  denote the subgraph of $G$ induced by $S$ and 
let  $G-S=G[V(G)\setminus S]$.  
{Given two disjoint vertex sets $V_1$ and $V_2$, the {\it join} $V_1\vee V_2$ is the edge set obtained by joining each vertex of $V_1$ to each vertex of $V_2$. Furthermore,  
	given two graphs $G$ and $H$,
	the {\it join} $G\vee H$ is the graph obtained from the disjoint union of $G$ and $H$  by joining each vertex of $G$ to each vertex of $H$. }
Let $r,t,k$ be three positive integers. For convenience, we write  $[r]=\{1,2,3,...,r\}$ in the context. 
Denote by $K_t$  the complete graph on $t$ vertices
 and  $kK_t$ the disjoint union of  $k$ copies of $K_t$. Moreover, for 
$r$ positive integers $n_1,n_2,\dots, n_r$, we use $K_{n_1,n_2,\dots,n_r}$ to denote the complete $r$-partite graph with parts of sizes $ n_1,n_2,\dots,n_r$.

Let {ex}$(G,H)=\max\{e(G'):$  $|V(G')|=|V(G)|$, $G'\subseteq G$, and $ H\nsubseteq G'\}$ and call it the {\it extremal number} of $H$ in $G$. Let $n$ and $t$ be two integers with $n\ge t$. 
 The classical Tur\'{a}n problem considers the  case $G= K_n$, i.e.,  determining the value of {ex}$(K_n,H)$ for a given $H$. For instance, the well-known Mantel's Theorem \cite{MW1907} and Tur\'{a}n's Theorem \cite{TP1941} determined {ex}$(K_n,K_3)$ and {ex}$(K_n,K_t)$, respectively.
  Tur\'{a}n's Theorem is {widely} considered to be the first extremal theorem {on graphs} and {initiated} the study of extremal graph theory. {Let~
  $T_r(n)$ be} a balanced complete $r$-partite
 graph on $n$ vertices, that is, any two parts of sizes
 differ by at most one.  
 In 1959, Erd\H{o}s and Gallai \cite{EPGT1959} determined {ex}$(K_n, kK_2)$ for any positive integers $n$ and $k$. 
 Later, Erd\H{o}s  \cite{EP1962} proved {ex}$(K_n, kK_3)=e(K_{k-1}\vee T_{2}(n-k+1))$ for  $n > 400(k-1)^2$ and  Moon \cite{MWJ1968}  proved that for $n> \frac{9k}{2} + 4$. Moreover,
 Moon \cite{MWJ1968} (only when { $n-k + 1$ is divisible by $p$})
 and Simonovits \cite{S1968} showed that $K_{k-1}\vee T_{r-1}(n-k+1)$  is the unique extremal graph 
 containing no copy of  $k K_{r}$ for sufficiently
 large $n$. 
 {The extremal problem on multipartite graphs} was first considered by {Bollob\'{a}s}, Erd\H{o}s, and Straus \cite{BBEPSGE1974}. They  determined {ex}$(K_{n_1,n_2,\dots,n_r},K_t)$ for $r\ge t\ge 2$. Later,  
Chen, Li, and Tu \cite{CHLXTJ2009} determined {ex}$(K_{n_1,n_2},kK_2)$. Recently,
 De Silva, Heysse and Young \cite{SDJHKYM} proved  {ex}$(K_{n_1,n_2,\dots,n_r},kK_2)=(k-1)(n_2+n_3+\cdots+n_r)$ for $n_1\le n_2\le \dots\le n_r$.


 {To make it easier to state the  results in this paper, we define a function $f$.}
Let $x_1,x_2,\dots,x_r$ { be  positive integers and   $t\ge 2$ be an integer}. Note that $\{x_1, x_2,\dots,x_r\}$ is a multiset.
 For any subset $P\subseteq[r]$,  let $x_P:=\sum_{i\in P}x_i$.  Now, we define the desired function $f$ as follows.
For $t=2$, let $f_2(\{x_1,x_2,\dots,x_r\})= 0$.
For $t\ge 3$,  let
$f_t(\{x_1,x_2,\dots,x_r\})= \max \limits_{\mathcal{P}} ~\{\sum \limits_{1\le i<j\le t-1}x_{P_i}\cdot x_{P_j}\}$, where $\mathcal{P}=(P_1,P_2,\dots,P_{t-1})$ is a partition of $[r]$
{ into $t-1$ nonempty parts.}
 For convenience,
{we simply write $f_t(x_1,x_2,\dots,x_r)$ for $f_t(\{x_1,x_2,\dots,x_r\})$.} 
{Notice that  $f_r(x_1,x_2,\dots,x_r)=(\sum_{1\le i<j\le r}x_ix_j)-x_1x_2$} if $x_1\le x_2\le \dots\le x_r$.  

\begin{theorem}\label{partitekt}{\em\cite{BBEPSGE1974}}
	Let $n_1, n_2, \dots, n_r,t$ be  positive integers with $r\ge t\ge 2$. Then\\ {ex}$(K_{n_1,n_2,\dots,n_r},K_t)=f_t(n_1,n_2,\dots,n_r)$. 
\end{theorem}

By replacing the forbidden graph $K_t$ with $kK_t$, 
De Silva et al. 
 \cite{SDJHKKASAYM2018} considered a special case $t=r$ and 
  raised the problem about 
{ex}$(K_{n_1,n_2,\dots,n_r},kK_t)$ for $r>t$.
\begin{theorem}{\em\cite{SDJHKKASAYM2018}}\label{rparkkr}
	For any integers $r\ge 2$ and $1\le k\le n_1\le n_2\le \cdots\le n_r$, we have 
	\begin{align*}
		{ex}(K_{n_1,n_2, \dots, n_r}, kK_r)&=\sum_{1\le i<j\le r}n_in_j-n_1n_2+(k-1)n_2\\&{=}(k-1)(n-n_1)+
	{f_{r}(n_1-(k-1),n_2,\dots,n_r),}
	\end{align*}
 {where $n=\sum n_i$. }
\end{theorem}

\begin{problem}{\em\cite{SDJHKKASAYM2018}}\label{problem}
	Determine {ex}$(K_{n_1,n_2,\dots,n_r},kK_t)$ for $r>t$.
\end{problem}

Han and Zhao \cite{HJZY2022+} determined {ex}$(K_{n_1,n_2,n_3,n_4},kK_3)$ for a sufficiently large integer  $n_1$ with  $n_1\le n_2\le n_3\le n_4$. 

{ 
\begin{theorem}\label{4tk3}{\em\cite{HJZY2022+}}
Let $n_1, \dots, n_4$ be sufficiently large and 
$  n_1\le n_2\le n_3\le n_4$.	For any integer $k\ge 1$,  we have 
	\begin{align*}
ex(K_{n_1, n_2, n_3, n_4}, kK_3)&=(k-1)(n-n_1)+f_3(n_1-{(k-1)},n_2,n_3, n_4)\\&=
	\left\{
	\begin{array}{lcl}
	n_4(n_1+n_2+n_3)+(k-1)(n_2+n_3), &&{\text{if}~n_4> n_2+n_3};\\
	(n_1+n_4)(n_2+n_3)+(k-1)n_4, &&{\text{if}~n_4\le n_2+n_3}.\\
	\end{array}\right.	\end{align*}
\end{theorem}
}

In addition, Han and  Zhao \cite{HJZY2022+} proposed a conjecture on {ex}$(K_{n_1,n_2,\dots,n_r},kK_t)$. 
\begin{conjecture}{\em\cite{HJZY2022+}}\label{conj-kK_t}
Given three integers $k$, $r$, and $t$ with $ r>t\ge 3$ and $k\ge 2$, let $n_1,\dots,n_r$ be sufficiently large. For $I\subseteq [r]$, write $m_I:={\min}_{i\in I}n_i$. Given a partition $\mathcal{P}$ of $[r]$, let $n_{\mathcal{P}}:={\max}_{I\in \mathcal{P}}\{n_I-m_I\}$. Then
	\begin{align*}
	ex(K_{n_1,n_2,\dots,n_r},kK_t)=\mathop{max}\limits_{\mathcal{P}}\left\{(k-1)n_{\mathcal{P}}+\sum_{I\ne {I'}\in {\mathcal{P}}}n_I\cdot n_{I'}\right\},
	\end{align*}
	where the maximum is taken over all partitions $\mathcal{P}$ of $[r]$ into $t-1$ parts.  
\end{conjecture}



Actually, we  can obtain an equivalent statement of Conjecture \ref{conj-kK_t} as follows, equivalence will be proved in Section  \ref{fproperty}. 

\begin{conjecture}\label{conj-z} 
Given three integers $k$, $r$, and $t$ with $ r>t\ge 3$ and $k\ge 2$, 	let $n_1\le n_2\dots\le n_r$ be integers with  $\sum_{i=1}^{r}n_i=n$. 
	If $n_1$ is sufficiently large,
	then {ex}$(K_{n_1,n_2,\dots,n_r},kK_t)= (k-1)(n-n_1)+f_t(n_1-(k-1),n_2,\dots,n_r)$.
\end{conjecture}
 
 To support Conjecture \ref{conj-z}, 
  we provide a simple construction, which gives a lower bound for {ex}$(K_{n_1,n_2,\dots,n_r},kK_t).$

 \begin{theorem}\label{kk3lower}
 	Let $k$ and $t$ be two integers with $k\ge 2$ and $ t\ge 3$. 	Let $n_1,n_2,\dots,n_r$ be $r$ integers with {$n_1\le \min_{i\ge 2}  \{n_i\}$}, $ n_1\ge k$, $r\ge t$, and $\sum_{i=1}^{r}n_i=n$. 
 	Then {ex}$(K_{n_1,n_2,\dots,n_r},kK_t)\ge (k-1)(n-n_1)+f_t(n_1-(k-1),n_2,\dots,n_r)$. 
 \end{theorem}
 
 Furthermore, 
 with  Theorem \ref{kk3lower}, we  confirm  Conjecture \ref{conj-z} for $t=3$, $r\ge 4$, and $n_1+4k\le n_2$.

  \begin{theorem}\label{mainkk3}
 	Let $r,n_1,n_2,\dots,n_r$ be  integers with {$r\ge 4$, $10k-4\le n_1+4k\le  \min_{i\ge 2} \{n_i\}$} and $\sum_{i=1}^{r}n_i=n$. Then 
 	{ex}$(K_{n_1,n_2,\dots,n_r},kK_3)= (k-1)(n-n_1)+
 	f_3(n_1-(k-1),n_2,\dots,n_r)$. 
 \end{theorem}

 
 In Conjecture \ref{conj-z}, it is required that   ``$n_1$ is sufficiently large".
 But how large it can be?
 By aforementioned results and Theorem \ref{mainkk3},
 we conjecture that the  lower bound  for $n_1$ 
 is  $k$.

  \begin{conjecture}\label{conj-z2} 
  	Given three integers $k$, $r$, and $t$ with $ r\ge t\ge 2$ and $k\ge 1$, 	let $n_1,n_2,\dots,n_r$ be  integers with {$n_1\le \min_{i\ge 2} \{n_i\}$} and $\sum_{i=1}^{r}n_i=n$. 
{If~ $n_1\ge k$},  
  	then {ex}$(K_{n_1,n_2,\dots,n_r},kK_t)= (k-1)(n-n_1)+f_t(n_1-(k-1),n_2,\dots,n_r)$.
  \end{conjecture}

The rest of this paper is organized as follows. 
In Section \ref{fproperty},   {we first prove the equivalence of Conjectures \ref{conj-kK_t} and \ref{conj-z}. Later, }
some basic
properties of $f_3(n_1,n_2,\dots,n_r)$ are provided, which will be used frequently in the proof of Theorem \ref{mainkk3}. The proofs of Theorems \ref{kk3lower} and    \ref{mainkk3} are presented in Sections \ref{lower} and \ref{upper}, respectively. 

 {
\section{
Equivalence and Properties of $f_3(n_1,n_2,\dots,n_r)$}\label{fproperty}

\begin{proposition}\label{equi}
	Conjecture 1.5  and Conjecture 1.6 are equivalent.
\end{proposition}

\begin{proof}
	Let $ r>t\ge 3$,  $k\ge 2$,  $n_1\le n_2\dots\le n_r$. 
	{ Our goal is to prove}
	\begin{align}\label{equi}
		\mathop{max}\limits_{\mathcal{P}}\left\{(k-1)n_{\mathcal{P}}+\sum_{I\ne {I'}\in {\mathcal{P}}}n_I\cdot n_{I'}\right\}=(k-1)(n-n_1)+f_t(n_1-(k-1),n_2,\dots,n_r){ .}
	\end{align}

	{ First we show the $\ge $ direction of   (\ref{equi}).}
	Let $x_1=n_1-(k-1)$, $x_i=n_i$ for any $i\ge 2$, and $\mathcal{P}_0=(P_1, P_2, \dots, P_{t-1})$ be a partition of $[r]$ maximizing $f_t(x_1,x_2,\dots,x_r)$. 
	Assume  $1\in P_1$. 
	Hence 
	\begin{align*}
		&(k-1)(n-n_1)+f_t(n_1-(k-1),n_2,\dots,n_r)\\
		&=(k-1)(n_{P_1}-n_1)+(k-1)(n-n_{P_1})+\sum_{i\ne j\in [t-1]}x_{P_i}\cdot x_{P_j}\\
		&\le (k-1)n_{\mathcal{P}_0}+\sum_{i\ne j\in [t-1]}n_{P_i}\cdot n_{P_j}\\
		&\le \mathop{max}\limits_{\mathcal{P}}\left\{(k-1)n_{\mathcal{P}}+\sum_{I\ne {I'}\in {\mathcal{P}}}n_I\cdot n_{I'}\right\}.
	\end{align*}
	
	Now we prove { the $\le $ direction of   (\ref{equi})}. 
	For any given partition $\mathcal{P}=(P_1, P_2, \dots, P_{t-1})$ of $[r]$, let $\ell\in [t-1]$ such that $n_\mathcal{P}=n_{P_\ell}-m_{P_\ell}$ and $\alpha \in P_\ell$ such that $n_\alpha=m_{P_\ell}$. 
	
	Notice that 
	{ 
		\begin{align*}
		& (k-1)n_{\mathcal{P}}+\sum_{I\ne {I'}\in {\mathcal{P}}}n_I\cdot n_{I'} \\
		&= (k-1)(n_{P_\ell}-m_{P_\ell})+n_{P_\ell}\cdot(n-n_{P_\ell})+
		\sum_{I\ne {I'}\in {\mathcal{P}\setminus\{P_\ell\}}}n_I\cdot n_{I'} \\
		&= (k-1)(n-m_{P_\ell})+(n_{P_\ell}-(k-1))\cdot(n-n_{P_\ell})+
		\sum_{I\ne {I'}\in {\mathcal{P}\setminus\{P_\ell\}}}n_I\cdot n_{I'} \\
		&\le  (k-1)(n-n_{\alpha})+f_t(\{n_1, \cdots, n_r\}\cup \{n_{\alpha}-(k-1)\}\setminus\{n_{\alpha}\}) .
	\end{align*}
Next we assume that $\mathcal{P}$ be the partition maximizing the value $(k-1)n_{\mathcal{P}}+\sum_{I\ne {I'}\in {\mathcal{P}}}n_I\cdot n_{I'} $. 
It remains to show \begin{align}\label{alpha1}
	 &(k-1)(n-n_{\alpha})+f_t(\{n_1, \cdots, n_r\}\cup \{n_{\alpha}-(k-1)\}\setminus\{n_{\alpha}\}) \notag\\&\le  (k-1)(n-n_{1})+f_t(\{n_1, \cdots, n_r\}\cup \{n_{1}-(k-1)\}\setminus\{n_{1}\}).
\end{align}
If $n_{\alpha}=n_1$, then we are done. We may assume that $n_{\alpha}> n_1$.  We can see $1\notin P_\ell$. 
 We may assume $1\in {P}_1$ and $\ell\ne 1$. Now note that $m_{P_1} = n_1$ and
 $m_{P_\ell}=n_\alpha$.  Then $n_{P_\ell }-n_\alpha > n_{P_1}-n_1$ by $\ell\ne 1$. 
 Since $n_{\alpha}> n_1$ and $n_{P_\ell }-n_\alpha > n_{P_1}-n_1$, we can switch $1$ and $\alpha$ in the partition $\mathcal{P}$ and thus increases $(k-1)n_{\mathcal{P}}+\sum_{I\ne {I'}\in {\mathcal{P}}}n_I\cdot n_{I'} $, a contradiction. Thus the Inequation (\ref{alpha1}) holds. 
}
	\hfill
\end{proof}
}

 {Next  we provided} some properties of the function $f_3(n_1,n_2,\dots,n_r)$. 
\begin{proposition}\label{f3property}
	Let $n_1, n_2, \dots, n_r,\mu$ be  positive integers with $\mu\in [r]$ and $n_{\mu}-1\ge n_1+1$. 
 Then $f_3(\{n_1,n_2,\dots,n_r\}\setminus \{n_1,n_\mu\}\cup \{n_1+1,n_\mu-1\})
\le f_3(n_1,n_2,\dots,n_r)+n_\mu-(n_1+1)$. 
\end{proposition}
\begin{proof}
	Let $x_1=n_1+1$, $x_\mu=n_\mu-1$, and $x_j=n_j$ for any $j\in [r]\setminus \{1,\mu\}$. Then 
	$f_3(\{n_1,n_2,\dots,n_r\}\setminus \{n_1,n_\mu\}\cup \{n_1+1,n_\mu-1\})=f_3(x_1,x_2,\dots,x_r)$. 
	Let  $\mathcal{P}=(P_1,P_2)$  be   the partition of $[r]$ maximizing  $f_3(x_1,x_2,...,x_r)$. 
{Note that} $\sum_{j=1}^rx_j=\sum_{j=1}^rn_j$. 
	 If $\{1,\mu\}\subseteq P_1$ or $\{1,\mu\}\subseteq P_2$, then $x_{P_1}=n_{P_1}$ and $x_{P_2}=n_{P_2}$. Thus, $f_3(x_1,x_2,\dots,x_r)=x_{P_1}\cdot x_{P_2}=n_{P_1}\cdot n_{P_2}\le f_3(n_1,n_2,\dots,n_r)$. 
	 {Hence, 
without} loss of generality,	we may assume that $1\in P_1$ and $\mu\in P_2$. 
	 Observe that $x_{P_1}=n_{P_1}+1$, $x_{P_2}=n_{P_2}-1$, $x_{P_1\setminus\{1\}}=n_{P_1\setminus\{1\}}$, and $x_{P_2\setminus\{\mu\}}=n_{P_2\setminus\{\mu\}}$.  If $x_{P_1\setminus\{1\}}\ge x_{P_2\setminus\{\mu\}}$, then $n_{P_1\setminus\{1\}}\ge n_{P_2\setminus\{\mu\}}$, and we have
	 \begin{align*}
	 	f_3(x_1,x_2,\dots,x_r)=&x_{P_1}\cdot x_{P_2}\\
	 	=&(n_{P_1}+1)\cdot(n_{P_2}-1)\\
	 	=&n_{P_1}\cdot n_{P_2}-n_{P_1}+ n_{P_2}-1\\
	 	\le &f_3(n_1,n_2,\dots,n_r)-(n_{P_1\setminus\{1\}}+n_1)+(n_{P_2\setminus\{\mu\}}+n_{\mu})-1\\
	 	\le &f_3(n_1,n_2,\dots,n_r)+n_{\mu}-(n_1+1). 
	 \end{align*}

	 {Now} $x_{P_1\setminus\{1\}}< x_{P_2\setminus\{\mu\}}${,  equivalently}, $n_{P_1\setminus\{1\}}< n_{P_2\setminus\{\mu\}}$. Recall that $x_1=n_1+1\leq n_\mu -1=x_\mu$. If $n_1+1 < n_\mu -1$, then  $|(x_{P_1\setminus\{1\}}+x_\mu)-( x_{P_2\setminus\{\mu\}}+x_1)|<|(x_{P_1\setminus\{1\}}+x_1)-( x_{P_2\setminus\{\mu\}}+x_\mu)|=|x_{P_1}-x_{P_2}|$, which contradicts the choice of $\mathcal{P}$ that minimizes $|x_{P_1}-x_{P_2}|$, i.e., maximizes $x_{P_1}\cdot x_{P_2}$.
	 Thus, we obtain $n_1+1 = n_\mu -1$. This implies $x_{P_1}\cdot x_{P_2}=(x_{P_1\setminus\{1\}}+x_\mu)\cdot (x_{P_2\setminus\{\mu\}}+x_1)$. Then,
	 \begin{align*}
	 	f_3(x_1,x_2,\dots,x_r)=&x_{P_1}\cdot x_{P_2}\\
	 	=&(x_{P_1\setminus\{1\}}+x_\mu)\cdot (x_{P_2\setminus\{\mu\}}+x_1)\\
	 	=&(n_{P_1\setminus\{1\}}+n_\mu-1)\cdot (n_{P_2\setminus\{\mu\}}+n_1+1)\\
	 	=&(n_{P_1\setminus\{1\}}+n_\mu)\cdot (n_{P_2\setminus\{\mu\}}+n_1)+(n_{P_1\setminus\{1\}}+n_\mu)-(n_{P_2\setminus\{\mu\}}+n_1) -1\\
	 	< &f_3(n_1,n_2,\dots,n_r)+n_{\mu}-(n_1+1). 
	 \end{align*}
	 \hfill
\end{proof}

\begin{proposition}\label{f3property2}
	
	Let $n_1,n_2,\dots,n_r$ be $r$ integers with $n_1+2\le \min_{i\ge 2} \{n_i\}$.  For two indices $i,j\in [r]$ with $i\ne j$, let $\mathcal{P}=(P_1,P_2)$ be the partition  { attaining} $f_3(\{n_1,\dots,n_r\}\setminus\{n_i,n_j\}\cup \{n_i-1,n_j-1\})$.  Then 
	{ \begin{align*}
			f_3(\{n_1,\dots,n_r\}\setminus\{n_i,n_j\}\cup \{n_i-1,n_j-1\})
			\le& f_3(n_1,n_2,\dots,n_r)-\sum_{m=1}^{r}n_m\notag\\
			&+max\{n_1 +2,n_i -n_1+1\}
		\end{align*}}

%
\end{proposition}
\begin{proof}
	Let  $x_i=n_i-1$, $x_j=n_j-1$, and $x_\ell=n_\ell$ for any $\ell\in [r]\setminus\{i,j\}$. Then we know $f_3(\{n_1,\dots,n_r\}\setminus\{n_i,n_j\}\cup \{n_i-1,n_j-1\})=f_3(x_1,x_2,\dots,x_r)$ and the partition $\mathcal{P}$ { attains} $f_3(x_1,x_2,\dots,x_r)$.  Moreover, for any $\ell\in [r]\setminus\{1\}$, we have $x_\ell\ge n_\ell-1> n_1\ge x_1$ since $n_\ell\ge n_1+2$. 
	
	{If} $\{i,j\}\nsubseteq P_\zeta$ for any $\zeta\in [2]$, we have $x_{P_\zeta}=n_{P_\zeta}-1$ and $x_{P_{3-\zeta}}=n_{P_{3-\zeta}}-1$. Thus 
	\begin{align*}
		f_3(x_1, x_2, \dots,x_r)=&x_{P_{\zeta}}\cdot x_{P_{3-\zeta}}=(n_{P_{\zeta}}-1)\cdot (n_{P_{3-\zeta}}-1)\\
		=&n_{P_{\zeta}}\cdot n_{P_{3-\zeta}}-n_{P_{\zeta}}-n_{P_{3-\zeta}}+1\\
		\le& f_3(n_1,n_2,\dots, n_r)+1-\sum_{m=1}^rn_m. 
	\end{align*}
{	
	Assume}  $\{i,j\}\subseteq P_\zeta$ for some $\zeta\in [2]$. {Notice that} $x_{P_{\zeta}}=n_{P_{\zeta}}-2$ and $x_{P_{3-\zeta}}=n_{P_{3-\zeta}}$. Thus,   
	\begin{align}\label{ineq-2n}
		f_3(x_1, x_2, \dots,x_r)=&x_{P_{\zeta}}\cdot x_{P_{3-\zeta}}=(n_{P_{\zeta}}-2)\cdot n_{P_{3-\zeta}}\notag\\
		=&n_{P_{\zeta}}\cdot n_{P_{3-\zeta}}-2n_{P_{3-\zeta}}\notag\\
		\le& f_3(n_1,n_2,\dots, n_r)-2n_{P_{3-\zeta}}. 
	\end{align}
	
	{If}  $1\in {P_{\zeta}}$, we have $x_{P_{3-\zeta}}\ge x_{P_{\zeta}}-x_1$. Otherwise $ x_{P_{\zeta}}\cdot x_{P_{3-\zeta}}<x_{P_{\zeta}}\cdot x_{P_{3-\zeta}}+x_1\cdot(x_{P_{\zeta}}-x_1-x_{P_{3-\zeta}})=(x_{P_{\zeta}}-x_1)\cdot(x_{P_{3-\zeta}}+x_1)$, which  contradicts to the choice of $\mathcal{P}$. With $x_{P_{\zeta}}+x_{P_{3-\zeta}}=x_{\mathcal{P}}$, we know
	$$n_{P_{3-\zeta}}=x_{P_{3-\zeta}}\ge\frac{x_\mathcal{P}-x_1}{2}\ge \frac{\sum_{m=1}^rn_m-2-n_1}{2}.$$ 
	Thus, by the inequality (\ref{ineq-2n}),
	\begin{align*}
		f_3(x_1, x_2, \dots,x_r) \le& f_3(n_1,n_2,\dots, n_r)-2n_{P_{3-\zeta}}\\
		\le& f_3(n_1,n_2,\dots, n_r)+n_1+2-\sum_{m=1}^rn_m. 
	\end{align*}

	{It remains to consider the case $1\in {P_{3-\zeta}}$. 
 Recall $x_i>x_1$ and $i\in {P_{\zeta}}$}, we have  $x_{P_{3-\zeta}}-x_1\ge x_{P_{\zeta}}-x_i$. Otherwise 
	$ x_{P_{\zeta}}\cdot x_{P_{3-\zeta}}<x_{P_{\zeta}}\cdot x_{P_{3-\zeta}}+(x_i-x_1)\cdot(x_{P_{\zeta}}-x_i-x_{P_{3-\zeta}}+x_1)=(x_{P_{\zeta}}-x_i+x_1)\cdot(x_{P_{3-\zeta}}-x_1+x_i)$,  a contradiction to the choice of $\mathcal{P}$ again. 
	With $x_{P_{\zeta}}+x_{P_{3-\zeta}}=x_{\mathcal{P}}$, we know $$n_{P_{3-\zeta}}=x_{P_{3-\zeta}}\ge\frac{x_\mathcal{P}+x_1-x_i}{2}= \frac{\sum_{m=1}^rn_m-1+n_1-n_i}{2}.$$
	Thus, by the inequality (\ref{ineq-2n}),
	 \begin{align*}
		f_3(x_1, x_2, \dots,x_r)
		\le& f_3(n_1,n_2,\dots, n_r)-2n_{P_{3-\zeta}}\\
		\le &f_3(n_1,n_2,\dots, n_r)-n_1+n_i+1-\sum_{m=1}^rn_m.
	\end{align*}
\hfill
\end{proof}

%
%
%

\section{Proof of Theorem \ref{kk3lower}}\label{lower}


We shall provide a  construction to complete the proof.
Let $x_1=n_1-(k-1)$ and $x_i=n_i$ for any { $i\ne 1$}.
{Then} $f_t(n_1-(k-1),n_2,\dots,n_r)=f_t(x_1,x_2,\dots,x_r)$.
Let $\mathcal{P}=(P_1,P_2,\ldots,P_{t-1})$ be the partition of $[r]$ that { attains} $f_t(x_1,x_2,\dots,x_r)$.
Assume that $V_0$ is a set of   $(k-1)$ vertices and $V_i$ is a set of $x_i$ vertices for any $i\in [r]$ {such that $V_0$, $V_1$, \dots, $V_r$ are pairwise disjoint}. According to the partition $\mathcal{P}$ of $[r]$, we set  $V_{P_\ell}=\cup_{i\in P_\ell}V_i$ for each $\ell\in [t-1]$. So $V_{[r]}=\cup_{\ell\in [t-1]}V_{P_\ell}$ and $|V_{P_\ell}|=x_{P_\ell}$. 
Now we construct the graph $G$ as follows. 
Let $V(G)=V_0\cup V_{[r]}$ and 
$$E(G)=\{V_0\vee(\cup_{i=2}^{r}V_i)\}\cup \{\cup_{1\le \ell <j\le t-1}(V_{P_\ell}\vee V_{P_j})\}.$$ 
We see $|V(G)|=(k-1)+\sum_{i=1}^rx_i=\sum_{i=1}^rn_i=n$,  and 
\begin{align*}
|E(G)|=&(k-1)(n-n_1)+\sum_{1\le \ell <j\le t-1}x_{P_\ell}\cdot x_{P_j}\\
=&(k-1)(n-n_1)+
f_t(n_1-(k-1),n_2,\dots,n_r). \end{align*}
{We can see} $G$ is a subgraph of $K_{n_1,n_2,\dots,n_r}$ and $G-V_0$ is a complete $(t-1)$-partite graph. So any copy of $K_t$ in $G$   must contain at least one vertex of $V_0$. Since $|V_0|=k-1$, $G$ contains no copy of  $kK_t$. Thus {ex}$(K_{n_1,n_2,\dots,n_r},kK_t)\ge (k-1)(n-n_1)+f_t(n_1-(k-1),n_2,\dots,n_r)$.
This completes the proof of Theorem \ref{kk3lower}.
\hfill
\ENDproof
\section{Proof of Theorem \ref{mainkk3}}\label{upper}

Recall that, in this section, $r,n_1,n_2,\dots,n_r$ {are  integers} with   {$r\ge 4$, $10k-4\le n_1+4k\le  \min_{i\ge 2} \{n_i\}$},  and $\sum_{i=1}^{r}n_i=n$.
By Theorem \ref{kk3lower}, it suffices to prove  
\begin{align}\label{ineq-less}
ex(K_{n_1,n_2,\dots,n_r},kK_3)\le (k-1)(n-n_1)+
f_3(n_1-(k-1),n_2,\dots,n_r).
\end{align}

By way of contradiction,    we suppose that $k$ is the minimum positive integer such that  {ex}$(K_{n_1,n_2,\dots,n_r},kK_3)> (k-1)(n-n_1)+
f_3(n_1-(k-1),n_2,\dots,n_r)$.  Thus, there exists a graph $G$, as a spanning {subgraph} of $K_{n_1,n_2,\dots,n_r}$, containing no   copy of  $kK_3$ and satisfying
\begin{align}\label{ineq-G-edge} e(G)=ex(K_{n_1,n_2,\dots,n_r},kK_3)>(k-1)(n-n_1)+
f_3(n_1-(k-1),n_2,\dots,n_r). 
\end{align}
Moreover, we use $V_1,V_2,\dots,V_r$ to denote the  parts of $G$ , where $|V_i|=n_i$  for each $i\in [r]$. It is worth noting that the inequality (\ref{ineq-less}) holds for $k=1$ by Theorem \ref{partitekt}. Hence, we may assume  $k\ge 2$.

	\begin{claim}\label{abadvertex}
		$G-\{v\}$ contains at least one copy of $(k-1)K_3$
		for each $v\in V(G)$.  
	\end{claim}
\begin{proof}
	By contradiction, suppose that there is a  vertex  $v$ in $V(G)$ such that $G-\{v\}$ contains no copy of $(k-1)K_3$.
{Suppose} that $v\in V_\ell$ for some $\ell\in[r]$. Let 
 $n_\ell'=n_\ell-1$, $n_j'=n_j$ for   any { $j\ne \ell$},
 and {$n'=n-1$}.  Then {
 $10(k-1)-4\le n'_1+4(k-1)\le  \min_{i\ge 2} \{n'_i\}$.} 
 Moreover, $G-\{v\}\subseteq K_{n'_1,n'_2,\dots,n'_r}$ by the construction.
 Because $G-\{v\}$ contains no copy of $(k-1)K_3$,  we obtain
 $e(G-\{v\})\le ex(K_{n'_1,n'_2,\dots,n'_r},(k-1)K_3)$.
  By the minimality of $k$, we have 
  \begin{align}\label{badvep}
 e(G)=&\notag e(G-\{v\})+d_G(v)\\\notag
 \le& ex(K_{n'_1,n'_2,\dots,n'_r},(k-1)K_3)+d_G(v)\\
 \le&(k-2)(n'-n_1')+f_3(n_1'-(k-2), n_2',n_3',\dots,n_r')+d_G(v).
 \end{align}
 If $v\in V_1$, then $n_1'=n_1-1$ and $d_G(v)=|N_G(v)|\le |V(G)\setminus V_1|=n-n_1$ since $G\subseteq K_{n_1,n_2,\dots,n_r}$. By the inequality (\ref{badvep}), we have 
 \begin{align*}
 e(G)\le& (k-2)(n-1-n_1+1)+f_3(n_1-1-(k-2), n_2,n_3,\dots,n_r)+n-n_1\\
 =&(k-1)(n-n_1)+f_3(n_1-(k-1),n_2,n_3,\dots,n_r),
\end{align*}
which contradicts the inequality (\ref{ineq-G-edge}). Hence,  we may assume  $v\notin V_1$. Then $d_G(v)\le n-n_{\ell}$,  $n_{\ell}'=n_{\ell}-1$, and $n_i'=n_i$ for any { $i\ne \ell$}. 
Note that $n_{\ell}-1\ge n_1+4k-1\ge n_1-k+2$. 
By the inequality (\ref{badvep}) and Proposition \ref{f3property}, we have \begin{align*}
e(G)\le& (k-2)(n-1-n_1)+f_3(\{n_1-(k-2), n_2,n_3,\dots,n_r\}\setminus\{n_{\ell}\}\cup \{n_{\ell}-1\})+n-n_{\ell}\\
\le&(k-2)(n-1-n_1)+(f_3(n_1-(k-1),n_2,n_3,\dots,n_r)+n_{\ell}-(n_1-k+2))+n-n_{\ell}\\
{=} &(k-1)(n-n_1)-k+2+n_1+f_3(n_1-(k-1),n_2,n_3,\dots,n_r)-(n_1-k+2)\\
=&(k-1)(n-n_1)+f_3(n_1-(k-1),n_2,n_3,\dots,n_r),
\end{align*}
which is a contradiction to the inequality (\ref{ineq-G-edge}).
This completes the proof of Claim \ref{abadvertex}. 
 \end{proof}

Since $k\geq2$, the graph $G$ contains at least one triangle by  Claim \ref{abadvertex}. Arbitrarily  choose a triangle $K^*_3$ in $G$ { with vertices $u_1,u_2,u_3$.}
Assume that 
$u_1\in V_\alpha$, $u_2\in V_\eta$, and $u_3\in V_\xi$, where $\alpha,\eta,\xi$ are three distinct integers in $ [r]$. Moreover, let $S=\{\alpha,\eta,\xi\}$, $n'_\ell=n_\ell-1$ for any $\ell\in S$, $n_j'=n_j$ for any $j\in [r]\setminus S$, and $n'=\sum_{i=1}^{r}n'_i$.  {Notice that} $n'=\sum_{i=1}^{r}n_i-3=n-3$ and $G-\{u_1,u_2,u_3\}\subseteq K_{n_1', n_2',\dots,n_r'}$.

%

\begin{claim} \label{Claim-oneK3} $e(G)\le(k-2)(n'-n'_1)+f_3(n'_1-(k-2),n_2',\dots,n_r')+\sum_{1\le i<j\le 3}|N_{G}(u_i)\cap N_{G}(u_j)|+n-3$.
\end{claim}
\begin{proof}
{Note} that $G-\{u_1,u_2,u_3\}$ contains no copy of $(k-1)K_3$ since $G$ contains no copy of $kK_3$.  Moreover, we have {$10(k-1)-4\le n'_1+4(k-1)\le  \min_{i\ge 2} \{n'_i\}$ }.
 By the minimality of $k$,  we know
 $e(G-\{u_1,u_2,u_3\})\leq ex(K_{n_1',n_2',\dots, n_r'},(k-1)K_3)\leq(k-2)(n'-n'_1)+f_3(n'_1-(k-2),n_2',\dots,n_r')${.}
  Thus, 
 \begin{align}\label{triep}
 	e(G)=&\notag e(G-\{u_1,u_2,u_3\})+\sum_{i=1}^{3}|N_{G}(u_i)|-3\\
 	\le& (k-2)(n'-n'_1)+f_3(n'_1-(k-2),n_2',\dots,n_r')+\sum_{i=1}^{3}|N_{G}(u_i)|-3.
 \end{align}
By the Principle of Inclusion-Exclusion, 
\begin{align}\label{tridegree}
\sum_{i=1}^{3}|N_{G }(u_i)|=&\notag|\bigcup_{i=1}^{3}N_{G}(u_i)|+\sum_{1\le i<j\le 3}|N_{G}(u_i))\cap N_{G}(u_j)|-|\bigcap_{i=1}^3N_{G}(u_i)|\\\notag
\le &|V(G)|+\sum_{1\le i<j\le 3}|N_{G}(u_i)\cap N_{G}(u_j)|\\
= &n+\sum_{1\le i<j\le 3}|N_{G}(u_i)\cap N_{G}(u_j)|. 
\end{align}
Combining inequalities (\ref{triep}) and (\ref{tridegree}), we see
\begin{align*}
e(G)\le (k-2)(n'-n'_1)+f_3(n'_1-(k-2),n_2',\dots,n_r')+\sum_{1\le i<j\le 3}|N_{G}(u_i)\cap N_{G}(u_j)|+n-3{.}
\end{align*}
\end{proof}

{We finish our proof in the following two cases.

{\noindent \bf Case 1. $|N_{G}(u_i)\cap N_{G}(u_j)|\leq 6(k-1)$ for every two distinct vertices $u_i$ and $u_j$  in $\{u_1,u_2,u_3\}$.}}



 
Recall that  $G-\{u_1,u_2,u_3\}\subseteq K_{n_1', n_2',\dots,n_r'}$.  We shall consider the two cases when $1\in \{\alpha,\eta,\xi\}$ and $1\notin \{\alpha,\eta,\xi\}$.
First, assume  $1\in \{\alpha,\eta,\xi\}$. Without loss of generality, say $1=\alpha$. Then $n_1'=n_1-1$ and   $\{n_2',n_3',\dots,n_r'\}=\{n_2,n_3,\dots,n_r\}\setminus\{n_\eta,n_\xi\}\cup\{n_\eta-1,n_\xi-1\}$. 
By Claim \ref{Claim-oneK3},
we have \begin{align}  
e(G)\le&\notag (k-2)(n-3-n_1+1)+18(k-1)+n-3\\\notag&+f_3(\{n_1-1-k+2,n_2,n_3,\dots,n_r\}\setminus\{n_\eta,n_\xi\}\cup\{n_\eta-1,n_\xi-1\})\\\notag
=&(k-1)(n-n_1)-2(k-2)-(n-n_1)+18(k-1)+n-3\\\notag
&+f_3(\{n_1-k+1,n_2,n_3,\dots,n_r\}\setminus\{n_\eta,n_\xi\}\cup\{n_\eta-1,n_\xi-1\})\\\notag
=&(k-1)(n-n_1)+n_1+16(k-1)-1\\&+f_3(\{n_1-k+1,n_2,n_3,\dots,n_r\}\setminus\{n_\eta,n_\xi\}\cup\{n_\eta-1,n_\xi-1\}).\notag 
\end{align}

{
Now, we assume  $1\notin \{\alpha,\eta,\xi\}$. Then 
$n_i'=n_i-1> n_1$ for any $i\in  \{\alpha,\eta,\xi\}$, and $n_j'=n_j$ for any $j\in [r]\setminus\{\alpha,\eta,\xi\}$.  By Claim \ref{Claim-oneK3}, 
\begin{align}\label{goodtrino}
	e(G)\le&\notag (k-2)(n-3-n_1)+18(k-1)+n-3\\&+\notag f_3(\{n_1-k+2,n_2,n_3,\dots,n_r\}\setminus\{n_\alpha,n_\eta,n_\xi\}\cup \{n_\alpha-1,n_\eta-1,n_\xi-1\})\\\notag
	=&(k-1)(n-n_1)+15(k-1)+n_1\notag
	\\ &+f_3(\{n_1-k+2,n_2,n_3,\dots,n_r\}\setminus\{n_\alpha,n_\eta,n_\xi\}\cup \{n_\alpha-1,n_\eta-1,n_\xi-1\}).
\end{align} 

Note that $n_\alpha -1\ge n_1-k+2$. 
By Proposition \ref{f3property}, 
\begin{align}\label{6kno11}
	f_3&\notag(\{n_1-k+2,n_2,n_3,\dots,n_r\}\setminus\{n_\alpha,n_\eta,n_\xi\}\cup \{n_\alpha-1,n_\eta-1,n_\xi-1\})\\&\le
	f_3(\{n_1-k+1,n_2,n_3,\dots,n_r\}\setminus\{n_\eta,n_\xi\}\cup \{n_\eta-1,n_\xi-1\})+n_\alpha-(n_1-k+2).
\end{align}
Thus, combining inequalities (\ref{goodtrino}) and (\ref{6kno11}), 
\begin{align} \label{goodtrifinal}
	e(G)\le& \notag (k-1)(n-n_1)+15(k-1)+n_1+n_\alpha-(n_1-k+2)\\\notag
	&+f_3(\{n_1-k+1,n_2,n_3,\dots,n_r\}\setminus\{n_\eta,n_\xi\}\cup \{n_\eta-1,n_\xi-1\})\\\notag
	=&(k-1)(n-n_1)+16(k-1)+n_\alpha-1\\
	&+f_3(\{n_1-k+1,n_2,n_3,\dots,n_r\}\setminus\{n_\eta,n_\xi\}\cup \{n_\eta-1,n_\xi-1\}).
\end{align}

In both cases ($1\in \{\alpha,\eta,\xi\}$ and $1\notin  \{\alpha,\eta,\xi\}$), we have inequality (\ref{goodtrifinal}) holds. 
Since $n_1-k+1+2\leq  \min_{i\ge 2} \{n_i\}$,
by Proposition \ref{f3property2} } and the inequality (\ref{goodtrifinal}), 
\begin{align*}
e(G)\le& (k-1)(n-n_1)+n_{\alpha}+16(k-1)-1+f_3(n_1-k+1,n_2,n_3,\dots,n_r)-(n-k+1)\\
&+{\max\{n_1-(k-1)+2, n_{\eta}-n_1+(k-1)+1\}}\\
=&(k-1)(n-n_1)+f_3(n_1-k+1,n_2,n_3,\dots,n_r)\\
&+{17k-18-n+\max\{n_{\alpha}+n_1-k+3, n_{\eta}+n_{\alpha}-n_1+k\}}.
\end{align*}
Since  {$n-n_{\alpha}-n_1\ge (r-2)n_2\ge 2n_2\ge 20k-8$ and $n+n_1-n_{\eta}-n_{\alpha}\ge (r-1)n_1\ge 3n_1\ge 18k-12$,} we have 
$$e(G)\le (k-1)(n-n_1)+f_3(n_1-k+1,n_2,n_3,\dots,n_r),$$
which contradicts  the inequality (\ref{ineq-G-edge}).

{
{\noindent \bf Case 2. There exist two distinct vertices $u_i$ and $u_j$  in $\{u_1,u_2,u_3\}$ such that 
	$|N_{G}(u_i)\cap N_{G}(u_j)|> 6(k-1)$.}
}

Without loss of generality, we may assume $|N_{G}(u_2)\cap N_{G}(u_3)|\ge 6(k-1)+1$. Then we have the following claim. 
\begin{claim}\label{findtri6k} 
	For  any  vertex set $A\subseteq V(G)$, if $|N_G(u_2)\cap N_G(u_3)\cap A|\ge 6(k-1)+1$, then there is a vertex $u_0$ with  $u_0\in A\cap N_G(u_2)\cap N_G(u_3)$ such that $|N_G(u_0)\cap N_G(u_i)|\le 3(k-1)$ for any $i\in \{2,3\}$. 
\end{claim}
\begin{proof}
By Claim \ref{abadvertex},
denote by $\mathcal{T}_i$ a  copy of $(k-1)K_3$ in $G\setminus\{u_i\}$ for any $i\in \{2,3\}$. Since $|N_G(u_2)\cap N_G(u_3)\cap A|\ge 6(k-1)+1$ and $|\mathcal{T}_1|=|\mathcal{T}_2|=3(k-1)$, there is a vertex, say $u_0$, such that $u_0\in  (N_G(u_2)\cap N_G(u_3)\cap A)\setminus (\mathcal{T}_1\cup \mathcal{T}_2)$. 
If $|N_{G}(u_0)\cap N_{G}(u_i)|\ge 3(k-1)+1$ for some $i\in \{2,3\}$, then there exists a vertex $w_i\in (N_{G}(u_0)\cap N_{G}(u_i))\setminus \mathcal{T}_i$. {Notice that} $\{w_i,u_0,u_i\}\cap \mathcal{T}_i=\emptyset$ and $u_0u_iw_iu_0$ is a triangle. So the disjoint union of $u_0u_iw_iu_0$ and $ \mathcal{T}_i$ is a copy of $kK_3$ in $G$, a contradiction to the choice of $G$.  Thus, $|N_G(u_0)\cap N_G(u_i)|\le 3(k-1)$ for any $i\in \{2,3\}$. 
\end{proof}

By setting $A=V(G)$ in Claim \ref{findtri6k}, there exists a vertex $u_0$ with  $u_0\in   N_G(u_2)\cap N_G(u_3)$ such that $|N_G(u_0)\cap N_G(u_i)|\le 3(k-1)$ for any $i\in \{2,3\}$.  Note that $u_0u_2u_3u_0$ is a triangle in $G$.  
Recall  $u_2\in V_\eta$ and $u_3\in V_\xi$. Assume $u_0\in V_{\alpha'}$.

According to the arbitrariness of the choice  of $K^*_3$ in $G$, we know that Claims \ref{Claim-oneK3} 
and \ref{findtri6k} are applicable to the triangle $u_0u_2u_3u_0$ by replacing $K^*_3$ with $u_0u_2u_3u_0$ and  replacing $S$ with $\{\alpha',\eta,\xi\}$. The remaining proof is divided into two subcases.


 {\bf Subcase 1. $ 1\in \{\alpha',\eta,\xi\}$.}

{Note that} $n_1'=n_1-1$ and $n'=n-3$. Moreover, we have $|N_G(u_2)\cap N_G(u_3)|\le n-n_\eta-n_\xi$ by $u_2\in V_\eta$ and $u_3\in V_\xi$. Let $\{1,s,\ell\}=\{\alpha',\eta,\xi\}$.
Due to $\{s,\ell\}\cap\{\eta,\xi\}\neq\emptyset$, without loss of generality, we may assume $s\in \{\eta,\xi\}$.
Since $|N_G(u_0)\cap N_G(u_i)|\le 3(k-1)$ for any $i\in \{2,3\}$,
by Claim \ref{Claim-oneK3},  we know

\begin{align*}
e(G)\le&(k-2)(n-3-n_1+1)+(6(k-1)+|N_G(u_2)\cap N_G(u_3)|)+n-3
\\
&+f_3(\{n_1-k+1,n_2,n_3,\dots,n_r\}\setminus\{n_s,n_\ell\}\cup \{n_s-1,n_\ell-1\})\\
\leq&  (k-1)(n-n_1)+4k-5+n+n_1-n_\eta-n_\xi
\\
&+f_3(\{n_1-k+1,n_2,n_3,\dots,n_r\}\setminus\{n_s,n_\ell\}\cup \{n_s-1,n_\ell-1\})\\
\le& {(k-1)(n-n_1)+4k-5+n+n_1-n_\eta-n_\xi +f_3(\{n_1-k+1,n_2,n_3,\dots,n_r\})}\\
& {-(n-(k-1))+\max\{(n_1 -k+1)+2,n_s -(n_1-k+1)+1\}}\\
\le & {(k-1)(n-n_1)+f_3(\{n_1-k+1,n_2,n_3,\dots,n_r\})}
\\& {+5k-6+n_1-n_\eta-n_\xi +\max\{n_1 -k+3,n_s -n_1+k+2\}}
\\
\le&  {(k-1)(n-n_1)+f_3(\{n_1-k+1,n_2,n_3,\dots,n_r\}),}
\end{align*}
 {where the third  inequality holds by Proposition \ref{f3property2}  and the last inequality holds by $4k-3+2n_1\le n_\eta+n_\xi$ and $6k-4+n_s\le n_\eta+n_\xi$. }

{\bf Subcase 2. $1\notin \{\alpha',\eta,\xi\}$.}

By setting $A=V_1$ in Claim \ref{findtri6k},  if $|N(u_2)\cap N(u_3)\cap V_1|\ge 6(k-1)+1$, then there exists a vertex $u'_0$ with  $u'_0\in    N_G(u_2)\cap N_G(u_3)\cap V_1$ such that $|N_G(u'_0)\cap N_G(u_i)|\le 3(k-1)$ for any $i\in \{2,3\}$.
Note that $u'_0u_2u_3u'_0$ is also a triangle in $G$.  This cases reduces to  Subcase 1.
 Hence, we may assume $|N(u_2)\cap N(u_3)\cap V_1|\le 6(k-1)$.  Then
\begin{align*}
|N_G(u_2)\cap N_G(u_3)|&=|N_G(u_2)\cap N_G(u_3)\cap (V(G)\setminus V_1)|+|N_G(u_2)\cap N_G(u_3)\cap  V_1|\\&\le n-n_\eta-n_\xi-n_1+6(k-1).
\end{align*} 
Recall that $n'=n-3$, $n'_1=n_1$, and $|N_G(u_0)\cap N_G(u_i)|\le 3(k-1)$ for any $i\in \{2,3\}$.
By Claim \ref{Claim-oneK3}, 
\begin{align}\label{ineq-fina-cla}
e(G)\le&\notag (k-2)(n-3-n_1)+(6(k-1)+n-n_\eta-n_\xi-n_1+6(k-1))+n-3\\\notag
&+f_3(\{n_1-k+2,n_2,n_3,\dots,n_r\}\setminus\{n_{\alpha'},n_{\eta},n_{\xi}\}\cup \{n_{\alpha'}-1,n_{\eta}-1,n_{\xi}-1\})\\\notag
=&\notag(k-1)(n-n_1)+9k-9+n-n_\eta-n_\xi\\&+f_3(\{n_1-k+2,n_2,n_3,\dots,n_r\}\setminus\{n_{\alpha'},n_{\eta},n_{\xi}\}\cup \{n_{\alpha'}-1,n_{\eta}-1,n_{\xi}-1\}).
\end{align}
Since $ n_1+4k\le  \min_{i\ge 2} \{n_i\}$, we know  $ n_1-k+3\leq\min_{i\ge 2} \{n_i\}$. By Proposition \ref{f3property} and the inequality (\ref{ineq-fina-cla}),
\begin{align}\label{badtrief}
	e(G)\le& \notag(k-1)(n-n_1)+9k-9+n-n_\eta-n_\xi+n_\xi-(n_1-k+2)\\\notag&+f_3(\{n_1-k+1,n_2,n_3,\dots,n_r\}\setminus\{n_{\alpha'},n_{\eta}\}\cup \{n_{\alpha'}-1,n_{\eta}-1\})\\\notag
	=&(k-1)(n-n_1)+10k-11+n-n_\eta-n_1\\&+f_3(\{n_1-k+1,n_2,n_3,\dots,n_r\}\setminus\{n_{\alpha'},n_{\eta}\}\cup \{n_{\alpha'}-1,n_{\eta}-1\}). 
\end{align}
{Proposition \ref{f3property2} }  implies that 
\begin{align}\label{ineq-fina-cor}
f_3&\notag(\{n_1-k+1,n_2,n_3,\dots,n_r\}\setminus\{n_{\alpha'},n_{\eta}\}\cup \{n_{\alpha'}-1,n_{\eta}-1\})\\
&\le f_3(n_1-k+1,n_2,\dots,n_r)+{\max\{n_1-(k-1)+2, n_{\eta}-(n_1-k+1)+1 \}-n+k-1.}
\end{align}
Since { $n_1\ge 6k-4$ and $n_\eta\ge n_1+4k\ge 10k-4$,}
combining  inequalities (\ref{badtrief}) and (\ref{ineq-fina-cor}), we have \begin{align*}
e(G)\le &(k-1)(n-n_1)+10k-11+n-n_\eta-n_1\\&\notag+f_3(n_1-k+1,n_2,\dots,n_r)+{\max\{n_1-k+3, n_{\eta}-n_1+k\}-n+k-1.}\\
=&(k-1)(n-n_1)+f_3(n_1-k+1,n_2,\dots,n_r)+11k-12\\
&- \max\{n_{\eta}+k-3, 2n_1-k\}\\
\le &(k-1)(n-n_1)+f_3(n_1-k+1,n_2,\dots,n_r), 
\end{align*}
which contradicts to the inequality (\ref{ineq-G-edge}).
This completes the proof of the inequality (\ref{ineq-less}), and so 
 Theorem \ref*{mainkk3} follows. 

\noindent{\bf Acknowledgements.}
The author would like to thank the editor and the two referees very much for their valuable suggestions and comments.
Junxue Zhang was partially supported by the National Natural Science Foundation of China (Nos. 12161141006 and 12111540249), the Natural Science Foundation of Tianjin (Nos. 20JCJQJC00090 and 20JCZDJC00840),  the Fundamental Research Funds for the Central Universities, Nankai University (No. 63223063).


\end{document}